\documentclass[a4paper,UKenglish,11pt]{article}
 
\usepackage{microtype}
\usepackage{verbatim}
\usepackage{amsmath}
\usepackage{amssymb}
\usepackage{amsthm}
\usepackage{graphicx}
\usepackage{color}
\PassOptionsToPackage{normalem}{ulem}
\usepackage{ulem}
\usepackage{url}

\usepackage{breakurl}
\usepackage{enumerate}
 \PassOptionsToPackage{normalem}{ulem}
 
\theoremstyle{plain}
\newtheorem{theorem}{Theorem}
\theoremstyle{plain}
\newtheorem{proposition}[theorem]{Proposition}
\theoremstyle{plain}
\newtheorem{lemma}[theorem]{Lemma}
\theoremstyle{plain}
\newtheorem{definition}[theorem]{Definition}

\newcommand{\IR}{\mathbb{R}}

\newcommand{\NP}{\textsf{NP}}

\newcommand{\DebatingLeague}{\textsf{DebatingLeague}}
\newcommand{\ThreeTeamDebating}{\textsf{ThreeTeamDebating}}
\newcommand{\BSAT}{\textsf{3-Bounded-3-SAT}}

\newcommand{\DebatingTournament}{\textsf{DebatingTournament}}

\title{This House Proves that Debating is Harder than Soccer}

\date{}

\author{Stefan Neumann\thanks{University of Vienna, Faculty of Computer Science, Vienna, Austria. The research leading to these results has received funding from the 
European Research Council under the European Union’s Seventh Framework Programme (FP/2007-2013) / ERC Grant Agreement no.\ 340506.}
	\and Andreas Wiese\thanks{Max Planck Institute for Computer Science, Saarbr\"ucken, Germany.}}

\begin{document}

\maketitle

\begin{abstract}
During the last twenty years, a lot of research was conducted on the
sport elimination problem: Given a sports league and its remaining
matches, we have to decide whether a given team can still possibly
win the competition, i.e., place first in the league at the end. Previously,
the computational complexity of this problem was investigated only
for games with two participating teams per game. In this paper we
consider Debating Tournaments and Debating Leagues in the British
Parliamentary format, where four teams are participating in each game.
We prove that it is \NP-hard to decide whether a given team can win
a Debating League, even if at most two matches are remaining for
each team. This contrasts settings like football where two teams play
in each game since there this case is still polynomial time solvable.
We prove our result even for a fictitious restricted setting with
only three teams per game. On the other hand, for the common setting of 
Debating Tournaments we show that this problem
is fixed parameter tractable if the parameter is the number of remaining
rounds $k$. This also holds for the practically very important question
of whether a team can still qualify for the knock-out phase of the
tournament and the combined parameter $k+b$ where $b$ denotes the
threshold rank for qualifying. Finally, we show that the latter problem
is polynomial time solvable for any constant $k$ and arbitrary values $b$ 
that are part of the input.
\end{abstract}

\section{Introduction}

Debating and soccer are deeply rooted in our society. Debating dates
back to the times of the ancient greek when already in 460 BC the
citizens of Athens were meeting in one of the first parliaments of
the world for discussions and votings \cite{AthenianDemocracy}.
This gave rise to the fine art of rhetoric, the skill to speak in
a public debate in a convincing manner, to give a solid argumentation
for the provided claims, and to win the support of the audience for
the own case. Since the ancient Greece the art of debating has
developed, and great speeches became milestones of history such as
the famous speech delivered by Martin Luther King on August~28, 1963
containing the dictum ``I have a dream''~\cite{Dream}.
Nowadays, all over the world there are debating societies at universities
and outside academia that are devoted to debates and public speaking.
This has a long tradition, for instance, the Cambridge Union Society
was founded in 1815 and has been run continuously for
more than 200 years now \cite{Cus}. Important for
this paper is that there are debating competitions: teams of
debaters meet and argue for and against the case of a previously specified
motion. The roles (pro and contra) are assigned randomly and thus
the debaters do not necessarily argue for the side that they personally
support.

Like debating, soccer is an integral part of the contemporary societies
in many countries. It is played by 250 million players in more than
200 countries which makes it the world's most popular sport \cite{FootballPopular}.
Even more people are passionate for watching the matches and supporting
their favorite teams. For instance, the final of the last world cup
2014 was watched by more than one billion people world wide \cite{Fifa}.

It is clear that debating and soccer play a significant role in modern
societies. However, one question has remained open: what is harder,
debating or soccer? Empirically there are only very few indications.
There are quotes by soccer players such as ``We lost because we didn't
win.'' (Ronaldo~\cite{Ronaldo}), ``I also told him that verbally.''%
~(Mario Basler~\cite{Basler}), ``It doesn't matter if it is Milano or Madrid as
long as it is Italy.''%
~(Andreas M\"oller~\cite{Moeller}), or ``I can see the carrot at the end of the tunnel.''
(Stuart Pearce~\cite{Pearce}) which suggest that excelling rhetorically
might be harder than playing soccer. On the other hand, the political
careers of heads of states typically surpass their soccer careers
by orders of magnitude. For instance, Gerhard Schr\"oder, the former
chancellor of Germany, played only in the Bezirksliga~\cite{Gerd}
which is nowadays the 7th level of the soccer league system in Germany.
For the current German chancellor Angela Merkel we are not aware of
any non-trivial soccer abilities. However, she is known to occasionally
frequent the German national team's changing room after important
matches~\cite{Merkel}.

From a scientific point of view it is difficult to compare debating
and soccer since they have only few intersection points that allow
a scientifically accurate comparison. One of the few is the following:
consider a league in which soccer/debating teams play matches against
each other according to a pre-defined schedule that indicates on which
match days which respective teams play each other. Consider your favorite team
$t_{1}$. The question is: are there outcomes for all remaining matches
such that $t_{1}$ wins the championship?

In soccer, this question is polynomial time solvable if there are
at most two remaining matches per team and \NP-hard for at most three
matches per team under the three-point rule~\cite{FootballElimination,FIFARules}.
The latter is nowadays ubiquitous in soccer leagues and tournaments
(such as in all FIFA world cups since~1994, in most national soccer
leagues since 1995, and in some of them even much earlier~\cite{CafeFutebol}).
It specifies that if a team wins a match it scores three points for
the league ranking and the losing team scores zero points, if the
match is a draw then both teams score one point.

For debating, we focus in this paper on the British parliamentary
style format that enjoys great popularity world wide and is played
for instance in the world universities debating championships \cite{DebatingRules}.
In this format, four teams are playing in each game
and the winning team scores three points, the second team scores
two points, the third team scores one point, and the fourth team scores
zero points. If in the final ranking multiple teams have the same number of points, then
a tiebreaker is used. For simplicity, in this paper we assume that
this tie-breaker is the total number of FUN papers written by members
of the team and that the team $t_{1}$ has written the most FUN papers
among all participating teams. Thus, $t_{1}$ wins the championship
if there is no team with more points than $t_{1}$ and the corresponding
problem is called \DebatingLeague{}.

\subsection{Our contribution}

In this paper we prove that \DebatingLeague{} is \NP-hard, even
if there are only two remaining matches to play for each team. This
shows that debating is computationally harder than soccer in two ways:
first, if there are only two remaining matches to play for each team
then in soccer we can decide in polynomial time whether a given team
can still win \cite{FootballElimination}. Secondly, for an arbitrary
number of remaining matches soccer is easy under the two-point rule
\cite{FootballElimination}, i.e., the winning team scores two points,
rather than three. The two-point rule has the important feature that
for each match there is a given number of points (two) that are completely
distributed among the participating teams. This is also the case in
debating: in each match there are six points available and they are
all distributed. While with this feature soccer is easy, \DebatingLeague{}
is \NP-hard despite of this which underlines the complexity of the
latter problem. To the best of our knowledge, this is the first time
that the elimination problem has been studied for games with more
than two teams per match. In fact, we prove that our hardness result
even holds in a fictitious setting in which only three teams participate
in a game and they score two, one, and zero points, respectively.

While \DebatingLeague{} is \NP-hard if only two matchdays are remaining,
we can show something different for the system that is typically played
in debating \emph{tournaments}. There, the matches of the teams are
defined in a similar way as in Swiss-system tournaments \cite{SwissRules} (which
are for instance common in chess): after each round the teams are
ordered according to the number of points they scored so far. Then
the teams ranked 1st-4th play one match, the teams ranked 5th-8th
play the second match, and so on. Since the pairings in each round
depend on the initial ranking and the outcomes of the previous rounds,
the above hardness result for \DebatingLeague{} does not apply.
In practice debating tournaments have a first phase organized as above and
a second phase that is played as a knock-out tournament. There is a
threshold $b$ specifying that the first $b$ teams of the final
ranking after the first phase qualify for the knock-out phase,
denoted as \emph{breaking}. A key question that a team typically asks itself
during a tournament is whether it can still break.
Formally, we denote by \DebatingTournament{} the problem of deciding whether $t_1$ can finish on place $b$ or better
with $k$ rounds left in the tournament.

We show that 
\DebatingTournament{} can be solved in time $O(f(k + b)\cdot n)$, i.e., the problem is fixed parameter tractable
for the combined parameter $k + b$. In particular, this implies that for any constant
$k$ it is polynomial time solvable to decide whether $t_1$ can win the tournament,
while for $k=2$ \DebatingLeague{} is \NP-hard. For our algorithm we first prove that if initially the
team $t_{1}$ is ``too far behind'', i.e., has a too large initial
rank depending on $k$ and $b$, then it cannot break anymore. For the remaining case
we provide an algorithm with a running time of $O(f(k +  b)\cdot n)$ for
a suitable function $f$. 
Additionally, we show that for constant $k$ the problem is polynomial
time solvable (for an arbitrary value of $b$ that is part of the
input). Thus, even for arbitrary $b$ the case that $k=2$ is in $\mathsf{P}$,
in contrast to \DebatingLeague.

\subsection{Other related work}

In 1966, Schwartz~\cite{PossibleWinners} proved that using flow networks
it can be decided in polynomial time whether a baseball team can still win a baseball league.
In baseball the winner of a game wins a single point and the looser gets zero points,
there is no tie.
McCormick~\cite{BaseballWithLosses} generalized this result by giving
a polynomial time algorithm which allowed
to fix a number of losses for the team that is supposed to win the league.
Wayne \cite{NewPropertyBaseball} characterised all teams of a baseball league which
can still win the league by giving a threshold value for the number of points
and the number of matches a team must have to be able to win the league.
He further gave a polynomial time algorithm to compute this threshold.
This result was later improved by Gusfield and
Martel~\cite{TheStructureAndComplexityOfSportsEliminationGames}
who gave thresholds for a bigger set of possible
outcomes of the matches. For baseball leagues they 
gave a faster algorithm to determine the threshold and further allowed leagues
with multiple divisions and wild-cards~\cite{TheStructureAndComplexityOfSportsEliminationGames}.

A major difference between baseball and soccer leagues is which
outcomes are possible in a single game. 
For soccer leagues with the
three-point-rule it was proven by \cite{FIFARules} and \cite{FootballElimination}
independently that it is \NP-hard to determine whether a team can win the league.
P\'alv\"olgyi~\cite{DecidingSoccerScores} proved
that when we are given the table of a soccer league and a list of games
that were played so far without their outcomes, it is \NP-hard
to decide whether this table is valid, i.e., whether the distribution
of points to the teams can be achieved by real outcomes of games.
In~\cite{RefiningSportsElimination}, the authors construct a hypergraph
representing the teams and their remaining matches. Depending on certain
properties of this graph they prove multiple hardness results for
the question whether a certain team can still win the competition.

In \cite{GeneralizedSportsCompetition}, Kern and Paulusma consider
games with two teams, but allow a game to have many different outcomes.
They prove that it can be decided in polynomial time whether a team can still
win the competition if and only if
in each match exactly $m$ points can be distributed
arbitrarily to both teams
(for any positive integer $m$). 

\section{Debating League}

In this section we prove that \DebatingLeague{} is \NP-hard, even
if each team has at most two remaining matches to play. First, let
us define the problem formally. Let $T=\{t_{1},\dots,t_{n}\}$ be
the set of teams participating in the debating league. We denote the
set of remaining matches by $M\subset T^{4}$, i.e., we have $(t_{i},t_{j},t_{k},t_{l})\in M$
iff the teams $t_{i},t_{j},t_{k}$ and $t_{l}$ still have to play
against each other in a match. We assume that each possible match
occurs at most once; further, throughout the whole section the
game schedule of remaining matches is fixed.
The winner of each match scores $3$ points,
the second placed team scores $2$ points, the third placed team scores
$1$ point and the loosing team does not get any point. We are given
a \emph{score vector} $s\in\IR^{n}$ with an entry $s_{i}$ for each
team $t_{i}$ that indicates how many points team $t_{i}$ already
obtained before playing the remaining matches. Notice that the tuple $(T,M,s)$
encodes all information we need about the competition. In the \DebatingLeague{}
problem we want to find out whether team $t_{1}$ can still win the
competition. %

\begin{definition}
In the \DebatingLeague{} problem we are given a tuple $(T,M,s)$
and we want to answer the question whether there are outcomes for
all matches $M$, such that at the end there is no team that has more
points than team $t_{1}$.
\end{definition}

We will prove that this problem is already \NP-hard when each team
has at most two remaining matches. We prove this first for a variant
of \DebatingLeague{} where we have only 3 teams per match and each
team has at most two matches left to play. In a game the winner gets
2 points, the second placed team gets 1 point and the looser gets
0 points. We still want to decide whether team $t_{1}$ can win the
competition. We denote this problem \ThreeTeamDebating. It can also
be characterised by a tuple $(T,M,s)$ similarly to above.

\begin{theorem}
\label{Thm:ThreeTeamDebatingNPHard} The \ThreeTeamDebating{} problem
is \NP-hard even when each team has at most two remaining matches to play. 
\end{theorem}

Before we start giving the proof of Theorem~\ref{Thm:ThreeTeamDebatingNPHard},
we introduce a  way to visualize instances of \ThreeTeamDebating{} as graphs.
Suppose we are given an instance $(T,M,s)$ of \ThreeTeamDebating{}
in which each team plays at most two matches. We visualize its matches
via a\emph{ game graph} $G=(V,E)$ in the following way: For each
game $g\in M$, we introduce a \emph{game vertex} $v_{g}\in V$. For
each team $t_{i}$ that participates in two matches $g,g'$, i.e.\ if
$t_{i}\in g$ and $t_{i}\in g'$, we introduce an edge $e_{i}$ connecting
$v_{g}$ and $v_{g'}$. Such an edge will be called a \emph{team edge}.
Each edge will receive a weight $w_{i}$ which encodes how many points
team $t_{i}$ can still get without obtaining more points than team
$t_{1}$. If a team has only one game remaining, we do not introduce
an edge for it. Notice that later team $t_1$ will not be part of the game
graph as we can assume w.l.o.g.\ that it wins all of its remaining games
and has no games left.

We prove Theorem~\ref{Thm:ThreeTeamDebatingNPHard} via a reduction
from \BSAT{} \cite{GareyJohnson} to \ThreeTeamDebating{}. Let $\varphi$ be a $\BSAT$
formula with variables $x_{1},\dots,x_{n}$ and clauses $C_{1},\dots,C_{m}$.
We can assume that each variable occurs in two or three different
clauses and that it occurs at least once positively and at least once
negatively. We can further assume that each clause has two or three literals.

We construct an instance $(T,M,s)$ of \ThreeTeamDebating. First,
we describe gadgets out of which our construction is composed and
prove some of their properties. Afterwards, we describe how to combine
the gadgets to the final instance. In the sequel, we will prove some
properties about our construction. We will use the term ``We can
assume that ...'' for the claim that team $t_{1}$ can still win
the championship if and only if it can still win the championship
for outcomes of the matches where the respective following statement
is true. In our construction, $t_{1}$ has no remaining game to play.
We distinguish the other teams into \emph{two-game teams} and \emph{one-game
teams}, where the former type has two remaining games to play and
the latter type has one remaining game to play. For each team, we
will define how many points it can still score without getting more
points than $t_{1}$. We will not exactly specify how many points
each team has initially since it matters only how many points it can
still get without overtaking $t_{1}$. 

For each variable $x$ in $\varphi$ we introduce a \emph{ring gadget}.
Assume that $x$ occurs in the three clauses $C_{i},C_{j},C_{k}$. The
ring gadget for a variable $x$ consists of the six games given by the set
$G_{x}:=\{g_{x,C_{i}}^{1},g_{x,C_{i}}^{2},g_{x,C_{j}}^{1},g_{x,C_{j}}^{2},g_{x,C_{k}}^{1},g_{x,C_{k}}^{2}\}$
and six teams two-game teams as specified by $T_{x}:=\{t_{x,C_{i}}^{1},t_{x,C_{i}}^{2},t_{x,C_{j}}^{1},t_{x,C_{j}}^{2},t_{x,C_{k}}^{1},t_{x,C_{k}}^{2}\}$.
If $x$ appears in only two clauses $C_{i},C_{j}$ we use the same
setup for a fictitious clause $C_{k}$.

The games of the teams in $T_{x}$ are visualized in Figure~\ref{fig:AssignedVariablecircle}.
Ignoring teams which are not from the set $T_{x}$ and which we will
introduce later, the game $g_{x,C_{i}}^{1}$ is played by the teams
$t_{x,C_{k}}^{2},t_{x,C_{i}}^{1}$, the game $g_{x,C_{i}}^{2}$ is
played by the teams $t_{x,C_{i}}^{1},t_{x,C_{i}}^{2}$, the game $g_{x,C_{j}}^{1}$
is played by the teams $t_{x,C_{i}}^{2},t_{x,C_{j}}^{1}$, the game
$g_{x,C_{j}}^{2}$ is played by the teams $t_{x,C_{j}}^{1},t_{x,C_{j}}^{2}$,
the game $g_{x,C_{k}}^{1}$ is played by the teams $t_{x,C_{j}}^{2},t_{x,C_{k}}^{1}$,
and the game $g_{x,C_{k}}^{2}$ is played by the teams $t_{x,C_{k}}^{1},t_{x,C_{k}}^{2}$.
Thus, when visualizing the games in $G_{x}$ and the teams in $T_{x}$
they form a cycle. Each team $g_{x,C_{\ell}}^{1}$ with $\ell\in\{i,j,k\}$
is allowed to get 2 points and each team $g_{x,C_{\ell}}^{2}$ with
$\ell\in\{i,j,k\}$ can get 3 points. The other teams participating
in the games $G_{x}$ (to be defined later) will only be able to score
exactly 1 point and hence they will not be able to win a game.
Hence, we can assume that in each game
$g\in G_{x}$ one team in $T_{x}$ that plays in $g$ must score 2
points. Furthermore, each team in $T_{x}$ can win at most one game
and since there are six games in $G_{x}$ and six teams in $T_{x}$,
each team in $T_{x}$ must win exactly one game.

One way to visualize the outcome of the circle games is to orient
each edge in the game graph. The team edge of a team $t\in T_{x}$ points towards
the unique game in which $t$ scores 2 points. In this viewpoint,
the following lemma implies that we can assume that all edges of the
cycle are either oriented clockwise or counter-clockwise.
\begin{proposition}
\label{pro:orientation-ring-gadgets} We can assume that either
the ring gadget is \emph{oriented clockwise}, i.e.\ game $g_{x,C_{\ell}}^{z}$
is won by team $t_{x,C_{\ell}}^{z}$ for $\ell\in\{i,j,k\}$ and $z\in\{1,2\}$,
or the ring gadget is \emph{oriented counter-clockwise}, i.e.\ game
$g_{x,C_{\ell}}^{2}$ for $\ell\in\{i,j,k\}$ is won by team $t_{x,C_{\ell}}^{1}$
and the games $g_{x,C_{i}}^{1}$,$g_{x,C_{j}}^{1}$,$g_{x,C_{k}}^{1}$
have winners $t_{x,C_{k}}^{2}$,$t_{x,C_{i}}^{2}$,$t_{x,C_{j}}^{2}$,
respectively. 
\end{proposition}

Later, the two possible orientations of the ring gadget for variable
$x$ will correspond to setting the variable $x$ to true or to false.
Next, we introduce a \emph{clause game} $g_{C}$ for each clause $C$
in $\varphi$. Let $C$ be a clause with variables $x,y,z$. We introduce
three two-game teams $t_{x,C}^{4},t_{y,C}^{4},t_{z,C}^{4}$ that play
$g_{C}$ and each of them will play in another game that we will define
later. Each of them can still score 2 points. Intuitively, the team
among them that scores 2 points in $g_{C}$ will correspond to the
variable that satisfies the clause $C$ in a satisfying assignment.
Note that for the names of the teams we do not distinguish whether
a variable $x$ occurs positively or negatively in $C$.

We describe now how we connect the clause games with the ring gadgets,
see Figure~\ref{fig:AssignedVariablecircle}. Let $x$ be a variable
that occurs in a clause $C$. For this occurrence, we introduced
the team $t_{x,C}^{4}$ above. We now introduce a game $g_{x,C}^{3}$,
a two-game team $t_{x,C}^{3}$, and a one-game team $t_{x,C,d}^{3}$.
The team $t_{x,C}^{3}$ can still get $1$ point and the team
$t_{x,C,d}^{3}$ can still get $2$ points. We define that $g_{x,C}^{3}$
is the second game of $t_{x,C}^{4}$, the only game of $t_{x,C,d}^{3}$,
and one of the two games that $t_{x,C}^{3}$ plays. The intuition
behind this construction is that if $t_{x,C}^{3}$ gets $0$ points
in its second game (that we have not specified yet) then the team
$t_{x,C}^{4}$ can score up to $2$ points in game $g_{C}$ (without
getting more points in total than $t_{1}$). On the other hand, if
$t_{x,C}^{3}$ gets $1$ point in its other game, then $t_{x,C}^{4}$
can score only up to 1 point in $g_{C}$ and in particular, it cannot
score 2 points in $g_{C}$ anymore. Later, the first case will correspond
to the case that $x$ satisfies $C$ whereas the second case will
correspond to the case that $x$ does not satisfy $C$.
\begin{proposition}
\label{pro:clauseGadget} Let $x$ be a variable appearing in a clause
$C$. We can assume that 
\begin{itemize}
\item if $t_{x,C}^{3}$ scores $1$ point in a game different than $g_{x,C}^{3}$
that it plays, then $t_{x,C}^{4}$ scores at most $1$ point in
game $g_{C}$, and 
\item if $t_{x,C}^{3}$ scores $0$ points in a game different than
$g_{x,C}^{3}$ that it plays, then $t_{x,C}^{4}$ can score
up to 2 points in the game $g_{C}$. 
\end{itemize}
\end{proposition}

We specify the second game for the team $t_{x,C}^{3}$ (i.e.,
the game different than $g_{x,C}^{3}$ that it plays). If $x$
appears positively in clause $C$ then this second game is defined
to be $g_{x,C}^{2}$, otherwise, this second game is defined to be
$g_{x,C}^{1}$. If the variable $x$ appears in three clauses $C_i, C_j, C_k$ then
three games from the $G_{x}$ are still missing one team, exactly
one per clause. For these games we add the one-game
teams $T_{x}^{d}:=\{t_{x,d_i},t_{x,d_{j}},t_{x,d_{k}}\}$,
each of them playing the game with the corresponding clause in the
subscript and each of them allowed to score 1 point. If $x$ appears
in only two clauses then we similarly add two one-game teams such
that each of them is allowed to score 1 point and by this we ensure
that each game in $G_{x}$ has three teams. This completes the definition
of the instance.
\begin{lemma}
If $\varphi$ is satisfiable then there is an outcome of the defined
instance of \ThreeTeamDebating{} such that no team gets more points
than $t_{1}$. \end{lemma}

\begin{proof}
Suppose we are given a satisfying assignment to the variables in $\varphi$.
From this satisfying assignment we will construct outcomes of the
games, such that team $t_{1}$ wins the championship. Intuitively,
we will want to assign all points as depicted in Figure~\ref{fig:AssignedVariablecircle}.
We will now describe this formally.

Let $x$ be a variable. If $x$ is true, then we orient the ring gadget
of $x$ counter-clockwise according to the left image in Figure~\ref{fig:AssignedVariablecircle};
formally, the winners of the games $G_{x}$ are assigned as defined
in Proposition~\ref{pro:orientation-ring-gadgets}. If $x$ is false,
then the ring gadget of $x$ is oriented clockwise according to the
right image in Figure~\ref{fig:AssignedVariablecircle}. All one-game
teams $T_{x}^{d}$ will place second in their games and thus obtain
a single point each.

Consider a clause $C$ with variables $x,y,z$. In the satisfying
assignment one of them must satisfy $C$. Assume w.l.o.g.\ that $x$
satisfies $C$. Then we let team $t_{x,C}^{4}$ score 2 points in
the game $g_{C}$ and we let an arbitrary team among $t_{y,C}^{4},t_{z,C}^{4}$
score 1 point and the other one 0 points. For the game $g_{x,C}^{3}$
we let the one-game team score 2 points and team $t_{x,C}^{3}$ score
1 point, team $t_{x,C}^{4}$ obtains no additional point from this
game. Team $t_{x,C}^{3}$ further scores $0$ points in its remaining
game (game $g_{x,C}^{1}$ or $g_{x,C}^{2}$, depending on whether
$x$ appears negatively or positively in $C$) and the remaining
point of this game goes to the team which can obtain 3 points in total.
In the game $g_{y,C}^{3}$ we let the team $t_{y,C}^{4}$ score $1$
point and team $t_{y,C}^{3}$ scores $1$ point from its game in $G_{y}$.
For the teams and games for variable $z$ we use the same distribution
of points as for $y$. We define the outcomes of the games in the
same way for each clause $C$. See Figure~\ref{fig:AssignedVariablecircle}
for a sketch of the outcomes described above.

\begin{figure}[!tbp]
\centering %
\begin{minipage}[b]{0.45\textwidth}%
 \includegraphics[width=\textwidth]{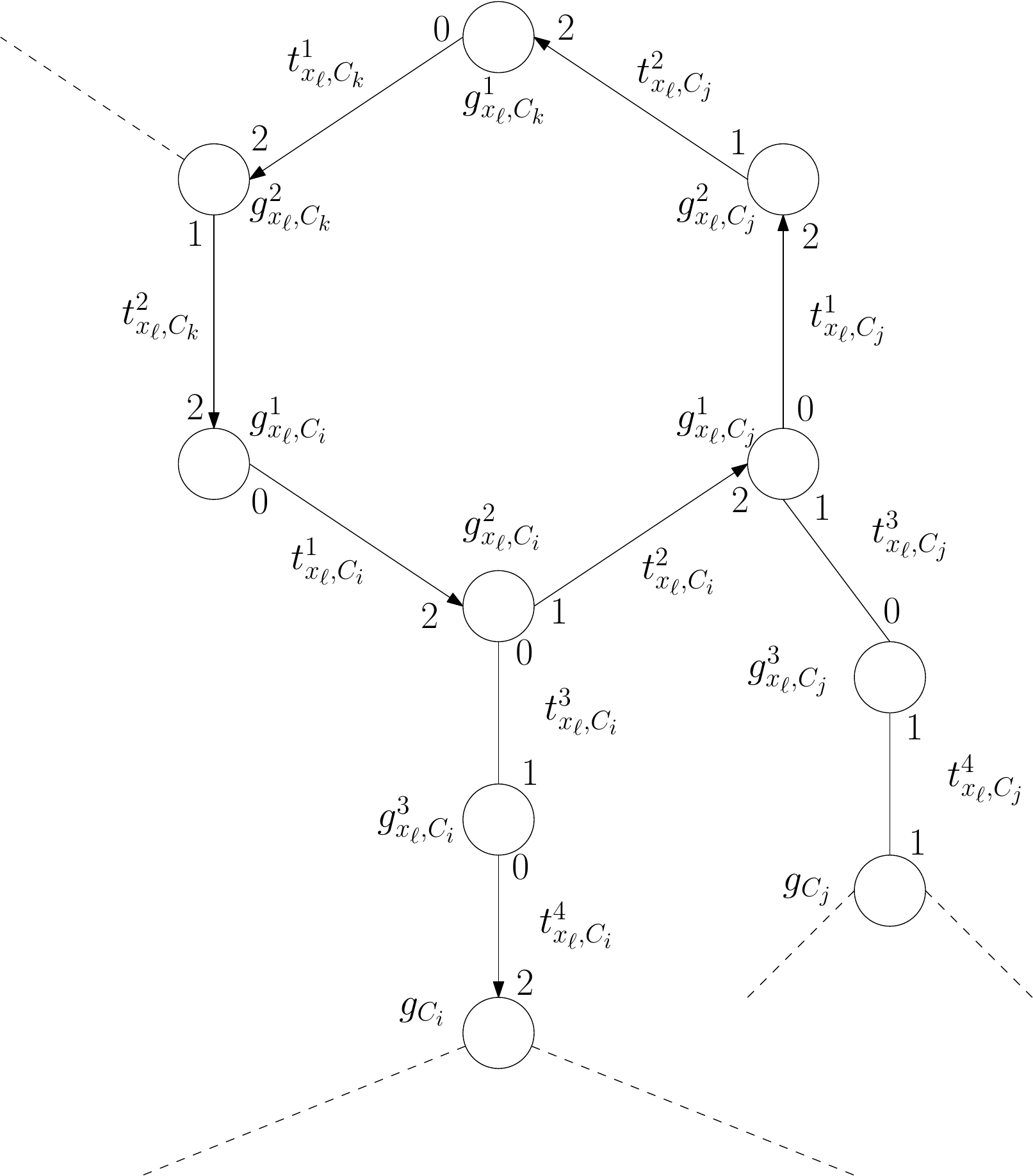} %
\end{minipage}\hfill{}%
\begin{minipage}[b]{0.45\textwidth}%
 \includegraphics[width=\textwidth]{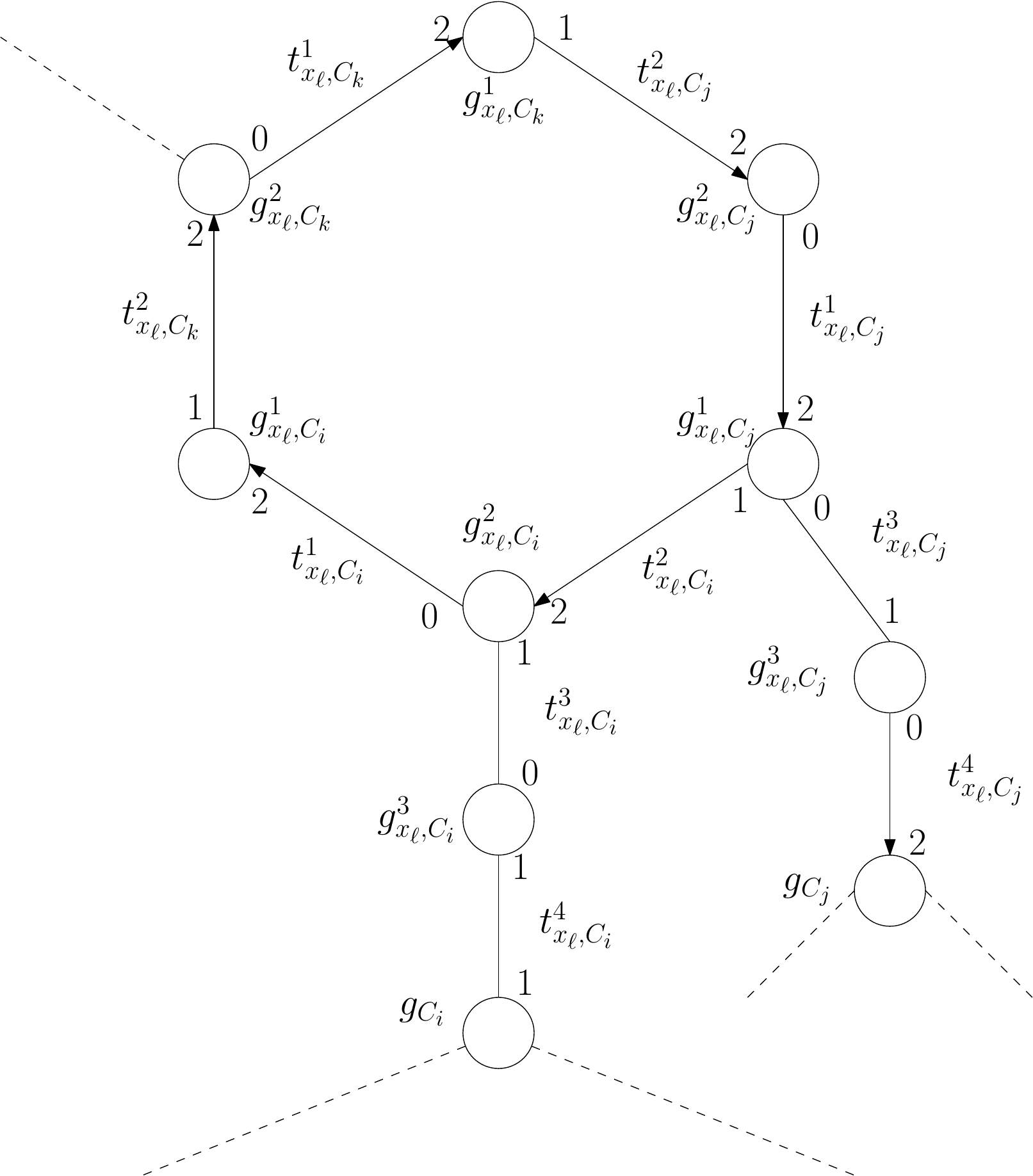} %
\end{minipage}\protect\caption{An excerpt of the game graph for variable $x_{\ell}$ which occurs
in clauses $C_{i}$ and $C_{k}$ positively and in $C_{j}$ negatively.
In the left image the outcomes of the games for $x_{\ell}=\emph{true}$ are
visualised, in the right image we have $x_{\ell}=\emph{false}$. The edges
are directed towards the game that was won by the corresponding team;
the numbers close to the game vertices show how many points the associated
two-game teams win in this game. }

\label{fig:AssignedVariablecircle} 
\end{figure}

All games distribute all of their points: In the above assignment,
for each clause game $g_{C}$ we have distributed all points by construction.
For all $x,C$, the games $g_{x,C}^{3}$ have a one-game team as a
winner and by construction the second place goes to either $t_{x,C}^{3}$
or $t_{x,C}^{4}$. It is left to to argue about the games from the
$G_{x}$. For each $g_{x,C}^{z}\in G_{x}$ with $z\in\{1,2\}$ we
must have a winner since we assigned the winners as defined in Proposition~\ref{pro:orientation-ring-gadgets}.
If $g_{x,C}^{z}$ has a one-game team participating then this team
can place second in our construction and hence all points are distributed.
If $g_{x,C}^{z}$ has only two-game teams, then we constructed our
variable assignment such that $g_{x,C}^{z}$ gives its last point
to $t_{x,C}^{3}$, if $x$ was not used to satisfy $C$. If $x$ was
used to satisfy $C$, then 1 point goes to its participating team
which can still get 3 points
(by the orientation
for the ring gadget we picked, this team cannot have won $g_{x,C}$).
Finally, by construction there is no team that scores more points
than we had specified, i.e., there is no team that scores more points
than $t_{1}$. %
\end{proof}

\begin{lemma}
If there is an outcome of the games in the defined instance of \ThreeTeamDebating{}
such that no team gets more points than $t_{1}$ then the formula
$\varphi$ is satisfiable. \end{lemma}

\begin{proof}
Suppose there is an outcome of the games such that no team gets more
points than $t_{1}$. We construct an assignment to the variables
in $\varphi$ that satisfies the formula. Let $x$ be a variable. Consider
the ring gadget for $x$. Due to Proposition~\ref{pro:orientation-ring-gadgets}
for the scores of the teams in $G_{x}$ there are two possibilities.
We set $x$ to be true if its ring gadget (as presented in Figure~\ref{fig:AssignedVariablecircle})
is oriented counter-clockwise in the sense of Proposition~\ref{pro:orientation-ring-gadgets},
otherwise, we set $x$ to false.

We prove that this variable assignment satisfies $\varphi$. Consider
a clause $C$ with three variables $x,y,z$. Assume w.l.o.g.\ that
$t_{x,C}^{4}$ scores 2 points in game $g_{C}$. We claim that then
$x$ satisfies $C$. Proposition~\ref{pro:clauseGadget} implies
that since $t_{x,C}^{4}$ scores 2 points in game $g_{C}$, team $t_{x,C}^{3}$
cannot get any points in game $g_{x,C}^{3}$. Hence, the other game
$g$ of $t_{x,C}^{3}$ must be won by a team which can get $2$ points
and have a second placed team which can achieve $3$ points.

If $x$ appears positively in $C$, then by construction we have $g=g_{x,C}^{2}$.
But then with the previous observation and Proposition~\ref{pro:orientation-ring-gadgets},
the ring gadget is oriented counter-clockwise and thus we have set $x$
to true. Thus, $x$ must satisfy $C$. On the other hand, if $x$
appears negatively in $C$, then we have $g=g_{x,C}^{1}$. This implies
that the ring gadget is oriented clockwise and thus we have set $x$
to false and hence $x$ satisfies $C$.
\end{proof}

Finally, we observe that in the above construction each team has at
most two remaining matches. This completes the proof of Theorem~\ref{Thm:ThreeTeamDebatingNPHard}.
Now we can show that \DebatingLeague{} is \NP-hard.

\begin{theorem}
  \label{Thm:DebatingLeagueNPHard}
  The \DebatingLeague{} problem is \NP-hard even if each team has
  at most two remaining matches to play.
\end{theorem}
\begin{proof}
Let $(T',M',s')$ be an instance of \ThreeTeamDebating{}. We modify
it to an instance of \DebatingLeague{}.
We begin by letting team $t_1$ win all of its remaining matches in the
\ThreeTeamDebating{} instance and updating the score vector
accordingly for all teams. In each of the games won by team $t_1$,
we replace $t_1$ by a dummy team that plays exactly
one match and can still score two points.
Now we update the instance to have four teams per match:
For each game $g$, we add
a dummy team that plays only in $g$ and that can still score
three points. Let $(T,M,s)$ denote the resulting instance of \DebatingLeague.
Observe that in this instance $t_1$ is not participating in any game.

If $(T',M',s') \in \ThreeTeamDebating$, then we can copy the outcomes
of all games to $(T,M,s)$ and then assign 3 points to each dummy team.
This gives a solution for \DebatingLeague.

On the other hand, consider an outcome of $(T,M,s)$ where $t_{1}$ wins
the championship.
Then the newly added dummy teams do not necessarily have to win their
respective games.
However, we can resolve this in the following way:
For each game won by a non-dummy team,
we change the outcome of the match such that the newly added dummy
team and the winning non-dummy team change positions.
Hence, a newly added dummy
team obtains 3 points and the other teams just get fewer points than
before. Thus, all newly added dummy teams win their respective games. 
We do a similar manipulation to make sure that the dummy teams that
replaced $t_1$ score exactly 2 points and we replace them by $t_1$. 
This implies that the outcomes of the matches disregarding the dummy teams give a solution
for $(T',M',s')$.
\end{proof}

We would like to point out that the above construction can easily
be adapted to show that it is also \NP-hard to decide whether $t_{1}$
can finish among the $b$ best teams for any constant $b$ (and thus
in particular if $b$ is part of the input). This can be achieved
by simply adding $b-1$ dummy teams that do not participate in any
game and initially have more points than $t_{1}$.

\section{Debating Tournaments}

In this section we will consider the \DebatingTournament{} problem:
We are given a set of teams $T=\{t_{1},\dots,t_{n}\}$, where $n$
is a multiple of $4$, and a vector $s\in\mathbb{R}^{n}$, where entry
$s_{i}$ specifies how many points team $t_{i}$ has scored so far.
We further get a parameter $k$ which indicates how many rounds (i.e.,
match days) are left to play. Contrary to the league setting from
the previous section, the fixtures are not determined beforehand.
At each match day the teams with ranks $4r+1,4r+2,4r+3$, and $4r+4$ for
each $r\in\mathbb{N}_{0}$ play a game. The points for winning the
games are distributed as in the \DebatingLeague{} setting. Additionally, we are given a parameter
$b$. We want to decide whether there are outcomes for all remaining matches such that 
at the end there are at most $b-1$ teams with more points than $t_{1}$. Since we assume
that in case of ties $t_{1}$ is always preferred, this means that $t_{1}$ finishes
among the $b$ best teams. This is an interesting question since in
debating tournaments it is common to have several rounds in the
above format, after which only the best $b$ teams are promoted to
the playoffs in which a knock-out elimination mode is played. Teams
who manage to finish among the $b$ best teams are said to \emph{break}.
Note that for $b=1$ this problem is identical to the question whether
team $t_{1}$ can still place first. We prove that the problem is fixed
parameter tractable (FPT) if both $k$ and $b$ are taken as parameters
by giving an algorithm with a running time of $O(f(k+b) \cdot n)$.

Recall the assumption that in tie-breaking $t_{1}$ is always preferred.
For the other teams, we assume w.l.o.g.~that we have a fixed
total order for the teams that specifies how to break ties if
two teams have exactly the same number of points. The next lemma states
a necessary condition for when $t_{1}$ can still break: $t_{1}$
has to be among the best $4^{k} b$ teams in the initial ranking
$s$. For our algorithm, we use this lemma to output ``no'' if $t_{1}$
is not among the first $4^{k}b$ teams in $s$. 
\begin{lemma}
\label{lem:FPTlemma} Let $t$ be a team that is among the best $4^{\ell}b$
teams when there are $\ell\in\{0,...,k\}$ rounds left to be played.
Then it has to be among the best $4^{\ell+1}b$ teams when there are
$\ell+1$ rounds left to be played. If a team is among the best
$b$ teams at the end of the tournament then it must be among the
best $4^{k}b$ teams when there are $k$ rounds left to be played.\end{lemma}

\begin{proof}
We start with the first claim. Assume for contradiction that team
$t$ is at a position larger than $4^{\ell+1}b$ when there are $\ell+1$
rounds left to be played and it is among the best $4^{\ell}b$ teams
when there are $\ell$ rounds left to be played. Observe that in the
round when $\ell+1$ games are left, the $4^{\ell+1}b$ best placed
teams will play $4^{\ell}b$ matches. Each of these games must have
a winner and among the participating teams $4^{\ell}b$ teams must win their
respective match (i.e., score 3 points) and thus will have more points
than $t$ when $\ell$ rounds are left, even if $t$ wins its match.
Hence, with $\ell$ rounds left to play, team $t$ must have a position
worse than $4^{\ell}b$.

The second claim can be shown by induction using the first claim as
the inductive step: If a team $t$ is among the best $b$ teams when
$\ell=0$ rounds are left to be played then it must be among the best
$4b$ teams before the last round, among the best $4^{2}b$ teams
before the last two rounds, \dots, and among the best $4^{k}b$ teams
when there are $k$ rounds left.
\end{proof}

Now we describe a recursive FPT algorithm with parameters $b$ and
$k$, that solves a given instance of \DebatingTournament{}. We
define two sets $S_{>t_{1}}:=\{t_{i}|s_{i}>s_{1}\}$ and $S_{\le t_{1}}:=\{t_{i}|s_{i}\le s_{1}\}$.
Both sets can be constructed in time $O(n)$. If $|S_{>t_{1}}|>4^{k}b$,
then the algorithm stops as team $t_{1}$ cannot break anymore by
Lemma~\ref{lem:FPTlemma}. Otherwise, the algorithm finds the best
$4^{k}b$ teams by taking team $t_{1}$, all teams from $S_{>t_{1}}$
and filling the remaining $4^{k}b-|S_{>t_{1}}| - 1$ slots with teams
from $S_{\leq t_{1}}$ in descending order of points. This step can
be implemented in time $O(4^{k}b \cdot n)$: we iterate over all elements
of $S_{\leq t_{1}}$ and keep track of the best team that was not yet added. 
When the iteration finished, we add the best team we found and mark it
as added. We have one iteration over $O(n)$ elements
for each free slot of of the $O(4^{k}b)$ teams,
and thus we need a running time of $O(4^{k}b \cdot n)$. Denote by $T^{(k)}$
the obtained set of teams.

The teams in $T^{(k)}$ play $4^{k-1}b$ matches. We guess the outcomes
of all these matches that still allow $t_{1}$ to be among the best
$b$ teams at the end. For each match there are $4!$ possible outcomes
and thus there are $(4!)^{4^{k-1}b}$ possible game outcomes to enumerate.
We update the scores of the teams accordingly. Denote by $T^{(k-1)}$
the first $4^{k-1}b$ teams in the resulting ranking. Lemma~\ref{lem:FPTlemma}
implies that in any outcome of all matches of the $n$ given teams
all teams in $T^{(k-1)}$ must also be in $T^{(k)}$. This justifies
that we enumerate only the matches for the teams in $T^{(k)}$, rather
than the matches for all $n$ given teams. Then we guess the outcome
of the $4^{k-2}b$ matches for the teams in $T^{(k-1)}$ that allows
$t_{1}$ to break eventually. We continue recursively for all remaining
rounds. For each guess of the outcomes of a round, e.g., when there
are only $\ell$ rounds remaining and we have $4^{\ell}b$ teams left
to consider, we make one recursive call to our routine with $\ell-1$
remaining rounds and $4^{\ell-1}b$ remaining teams. 

To evaluate the complexity of the algorithm let us observe that for
a single matchday there are at most $(4!)^{4^{k-1} b}$ possible outcomes,
since 
each match has $4!$ possible outcomes and during a single round of the
tournament there are at most $4^{k-1} b$ games to be played.
The recursion depth is $k$ which yields an overall running time of
$\left((4!)^{4^{k-1} b}\right)^{k} = 2^{(2b)^{O(k)}}$
of our algorithm.

In total, we need time $O(4^{k}b \cdot n)$ for the first phase of
the algorithm in which we determine the best $4^k b$ teams.
For the simulation of all possible outcomes we need time
$2^{(2b)^{O(k)}}$.
Note that if we set $b=1$, the algorithm decides in time $n\cdot2^{2^{O(k)}}$
whether $t_{1}$ can place first in a tournament without playoffs.
Thus, this problem is FPT for parameter $k$

\begin{theorem}
If there are $k$ remaining rounds to be played in a debating tournament,
there is an algorithm that decides in time $n\cdot2^{(2b)^{O(k)}}$
whether $t_{1}$ can place among the first $b$ teams at the end of
the tournament. 
\end{theorem}

\subsection{Constant number of rounds}

We present an algorithm that decides in time $n^{O(k^{4})}$ whether
a team can still break if there are $k$ more rounds to play. In particular,
this implies that for any constant $k$ the problem is polynomial
time solvable, in contrast to \DebatingLeague{}.

As before, suppose we are given a ranking with $n$ teams where for each team
$t_{i}$ we are given a value $s_{i}$ that denotes how many points
team $i$ has scored so far. Again, assume that after the last
round the first $b$ teams in the ranking break (and thus participate
in the play-offs). Also, we are given a value $k$ that denotes the
number of remaining rounds and we want to decide whether $t_{1}$
can still break. Consider a round such that including this round there
are only $\ell\le k$ more rounds to play. For each team $t_{i}$
let $s_{i}^{\ell}$ denote its score at the beginning of the round. We
distinguish three types of teams: teams $t_{i}$ with $s_{i}^{\ell}>s_{1}^{\ell}+3\ell$,
teams $t_{i}$ with $s_{1}^{\ell}-3\ell\le s_{i}^{\ell}\le s_{1}^{\ell}+3\ell$,
and teams $t_{i}$ with $s_{i}^{\ell}<s_{1}^{\ell}-3\ell$. Denote
those teams by $T_{T}^{\ell},T_{M}^{\ell},$ and $T_{B}^{\ell}$,
respectively (for \uline{t}op, \uline{m}iddle, and \uline{b}ottom).
At the end of the tournament, the final score for each team $t_{i}$
will be in $\{s_{i}^{\ell},...,s_{i}^{\ell}+3\ell\}$. Thus, during
the last $k$ rounds team $t_{1}$ cannot overtake any of the teams
in $T_{T}^{k}$ and none of the teams in $T_{B}^{k}$ can overtake
$t_{1}$. Thus, intuitively, only the exact scores teams in $T_{M}^{k}$
are relevant when deciding whether $t_{1}$ can still break. Our algorithm
enumerates all possible remaining outcomes of the remaining matches
but in doing so, it does not keep track of the scores of the teams
in $T_{T}^{k}\cup T_{B}^{k}$. For the initial scores of the teams
in $T_{M}^{k}$ there are only $O(k)$ possibilities and during $k$
rounds a team can score at most $O(k)$ points. Thus there are also only $O(k)$
possibilities for the scores of teams in $T_{M}^{k}$ during the last
$k$ rounds. In order to describe the ranking for those teams, up
to permutations it sufficies to keep track of the total number of
teams with each of the $O(k)$ possible scores. This yields $n^{O(k)}$ many
possibilities in total which allows us to solve the problem via a
dynamic program.

Formally, we will pretend that all teams in $T_{T}^{k}$ have exactly
the same number of points initially and that the same is true for
all teams in $T_{T}^{k}$. This is justified by the following lemma. 
\begin{lemma}
\label{lem:when-breaking}Assume that there are only $\ell$ rounds
left to play. Consider an initial ranking given by a number of points
$s_{i}^{\ell}$ for each team $t_{i}$. Then $t_{1}$ can still break
if and only if it can still break in any initial ranking given by
a number of points $\bar{s}_{i}^{\ell}$ for each team $t_{i}$ such
that 
\begin{itemize}
\item $s_{1}^{\ell}=\bar{s}_{1}^{\ell}$, 
\item there is a bijection $f:T_{M}^{\ell}\rightarrow\bar{T}_{M}^{\ell}:=\{t_{i}|\bar{s}_{i}^{\ell}-3\ell\le\bar{s}_{1}^{\ell}\le\bar{s}_{i}^{\ell}+3\ell\}$
such that for each $t_{i}\in T_{M}^{\ell}$ we have that in $s^{\ell}$
and $\bar{s}^{\ell}$ the teams $t_{i}$ and $f(t_{i})$ have the
same rank and the same scores and $f(t_{1})=t_{1}$, 
\item $|T_{T}^{\ell}|=|\bar{T}_{T}^{\ell}|:=|\{t_{i}|\bar{s}_{i}^{\ell}>\bar{s}_{1}^{\ell}+3\ell\}|$
and $|T_{B}^{\ell}|=|\bar{T}_{B}^{\ell}|:=|\{t_{i}|\bar{s}_{i}^{\ell}<\bar{s}_{1}^{\ell}-3\ell\}|$. 
\end{itemize}
\end{lemma}

\begin{proof}
We prove the claim by induction. For $\ell=0$ it is immediate since
$t_{1}$ can still break if and only if $|T_{T}^{0}|<b$. Suppose
now the claim is true for some value $\ell$ and we want to prove
it for $\ell+1$. It is immediate that $t_{1}$ can break in the initial
ranking $s^{\ell+1}$ if it can break in any initial ranking $\bar{s}^{\ell+1}$
with the above properties since $s^{\ell+1}$ satisfies these properties.

Now suppose that $t_{1}$ can break in the initial ranking $s^{\ell+1}$
and consider an initial ranking $\bar{s}^{\ell+1}$ with the above
properties. Consider the outcome of the games in the current round
for the ranking $s^{\ell+1}$ such that $t_{1}$ breaks after the
last round. We construct an outcome of the games of the current round
for the initial ranking $\bar{s}^{\ell+1}$. Consider a game $\bar{g}$
in which the teams $\{\bar{t}^{(1)},\bar{t}^{(2)},\bar{t}^{(3)},\bar{t}^{(4)}\}$
participate. There is a corresponding game $g$, played by team $\{t^{(1)},t^{(2)},t^{(3)},t^{(4)}\}$
according to the initial ranking $s^{\ell+1}$ such that for each
$j\in\{1,2,3,4\}$ we have that 
\begin{itemize}
\item if $\bar{t}^{(j)}\in\bar{T}_{M}^{\ell}$ then $t^{(j)}=f^{-1}(\bar{t}^{(j)})\in T_{M}^{\ell}$
and thus in $s^{\ell+1}$ and $\bar{s}^{\ell+1}$ the teams $\bar{t}^{(j)}$
and $t^{(j)}$ have exactly the same rank and the same score, 
\item if $\bar{t}^{(j)}\in\bar{T}_{T}^{\ell}$ then $t^{(j)}\in T_{T}^{\ell}$,
and 
\item if $\bar{t}^{(j)}\in\bar{T}_{B}^{\ell}$ then $t^{(j)}\in T_{B}^{\ell}$. 
\end{itemize}
Note that the first property implies that if $\bar{t}^{(j)}=t_{1}$
then $t^{(j)}=t_{1}$. For defining the outcome of $\bar{g}$ we simply
the take of the outcome of game $g$ from the known outcomes for all
remaining matches that let $t_{1}$ break eventually. For each $j\in\{1,2,3,4\}$
we assign the team $\bar{t}^{(j)}$ exactly the same score as team
$t^{(j)}$ in those outcomes. We do this operation with all games
$\bar{g}$. Denote by $\bar{s}^{\ell}$ the resulting ranking and
by $s^{\ell}$ the ranking resulting if we apply those outcomes to
$s^{\ell+1}$. Based on the induction hypothesis, we claim that $t_{1}$
can still break in $\bar{s}^{\ell}$. First, it is clear that $s_{1}^{\ell}=\bar{s}_{1}^{\ell}$
since $f(t_{1})=t_{1}$. Consider a team $t_{i}$. If $t_{i}\in\bar{T}_{T}^{\ell-1}$
then $t_{i}\in\bar{T}_{T}^{\ell}$ and also if $t_{i}\in T_{T}^{\ell-1}$
then $t_{i}\in T_{T}^{\ell}$. Similarly, if $t_{i}\in\bar{T}_{B}^{\ell-1}$
then $t_{i}\in\bar{T}_{B}^{\ell}$ and also if $t_{i}\in T_{B}^{\ell-1}$
then $t_{i}\in T_{B}^{\ell}$. Finally, if $t_{i}\in\bar{T}_{M}^{\ell-1}$
then 
\begin{itemize}
\item $t_{i}\in\bar{T}_{M}^{\ell}$ if and only if $f^{-1}(t_{i})\in T_{M}^{\ell}$
and then $t_{i}$ and $f^{-1}(t_{i})$ have the same score in $s^{\ell}$
and $\bar{s}^{\ell}$ 
\item $t_{i}\in\bar{T}_{T}^{\ell}$ if and only if $f^{-1}(t_{i})\in T_{T}^{\ell}$,
and 
\item $t_{i}\in\bar{T}_{B}^{\ell}$ if and only if $f^{-1}(t_{i})\in T_{B}^{\ell}$. 
\end{itemize}
Therefore, $|T_{T}^{\ell}|=|\bar{T}_{T}^{\ell}|$ and $|T_{B}^{\ell}|=|\bar{T}_{B}^{\ell}|$
and also there is a bijection $f:T_{M}^{\ell}\rightarrow\bar{T}_{M}^{\ell}$
with the properties required by the induction hypothesis. Thus, the
induction hypothesis implies that $t_{1}$ can still break when starting
with the initial ranking $\bar{s}^{\ell}$.
\end{proof}

We use Lemma~\ref{lem:when-breaking} to justify that we can work
with a new initial ranking $s'$ instead of $s$. Note that the sets
$T_{B}^{k}\dot{\cup}T_{M}^{k}\dot{\cup}T_{T}^{k}$ form a partition
of the participating teams. For each team $t_{i}\in T_{B}^{k}$ we
define $s'_{i}:=0$. For each team $t_{i}\in T_{M}^{k}$ we define
$s'_{i}:=s_{i}$. For each team $t_{i}\in T_{T}^{k}$ we define $s'_{i}:=s_{1}+3k+1$.
The next proposition follows immediately from Lemma~\ref{lem:when-breaking}. 
\begin{proposition}
The team $t_{1}$ can break with the initial ranking $s'$ if and
only if it can break with the initial ranking $s$. 
\end{proposition}

In our algorithm, we use a dynamic program in order to enumerate all
possible outcomes of the remaining $k$ rounds when starting with
the initial ranking $s'$. Key is that there are only $O(k)$ different
scores that a team can have during these $k$ rounds since there are
only $O(k)$ different initial scores and each team can score at most
$3k$ many points. We call two score vectors $\tilde{s},\tilde{s}'$
\emph{equivalent }if $\tilde{s}_{1}=\tilde{s}'_{1}$ and if for each
value $x$ the number of teams with exactly $x$ points is the same
in $\tilde{s}$ and $\tilde{s}'$. The team $t_{1}$ can clearly break
for an initial score vector $\tilde{s}$ if and only if it can still
break in any equivalent initial score vector $\tilde{s}'$. 
\begin{lemma}
When starting with the score vector $s'$, there are only $n^{O(k)}$
equivalence classes for the score vectors arising during the last
$k$ rounds. \end{lemma}

\begin{proof}
For the number of points of $t_{1}$ there are only $O(k)$ possibilities.
The other teams there can have at most $O(k)$ different scores. Thus,
in order to describe an equivalence class it suffices to specify the
points of $t_{1}$ and how many teams there are with each of the $O(k)$
possible different scores. This gives only $n^{O(k)}$ different possibilities
in total. 
\end{proof}

Our dynamic program works as follows: we have a DP-table entry $(\ell,C)$
for each $\ell\in\{0,...,k\}$ and each equivalence class $C$ of
the possibly arising score vectors. We store either ``yes'' or ``no''
in this cell, corresponding to whether or not $t_{1}$ can still break
if there are $\ell$ more rounds to play and we start with a score
vector that is equivalent to $C$. 
\begin{lemma}
Let $\ell\in\{0,...,k\}$. Suppose we have computed the entry of the
cell $(\ell,C')$ for each equivalence class $C'$. Then in time $n^{O(k^4)}$
we can compute the entry for a cell $(\ell+1,C)$.\end{lemma}

\begin{proof}
Consider a score vector corresponding to $C$. We distinguish the
different types of the games arising in the current round. We say
that two games with teams $\{t^{(1)},t^{(2)},t^{(3)},t^{(4)}\}$ and
$\{\bar{t}^{(1)},\bar{t}^{(2)},\bar{t}^{(3)},\bar{t}^{(4)}\}$, respectively,
are of the same \emph{type} if there exists a bijection
$g:\{t^{(1)},t^{(2)},t^{(3)},t^{(4)}\}\rightarrow\{\bar{t}^{(1)},\bar{t}^{(2)},\bar{t}^{(3)},\bar{t}^{(4)}\}$
such that $t^{(j)}$ and $g(t^{(j)})$ have exactly the same score
for each $j\in\{1,2,3,4\}$. There are only $O(k^{4})$ types of games
at only $4!$ different outcomes for each game. Thus, in order to
enumerate all possible outcomes of all games it suffice to guess how
many games of each type have which of the $4!$ possible outcomes.
Finally, there are $4!$ possible outcomes for the game that $t_{1}$
participates in. This gives $n^{O(k^{4})}$ possibilities in total
and for each possibility we obtain a cell $(\ell,C')$ for some equivalence
class $C'$. 
\end{proof}

Thus, in time $n^{O(k^{4})}$ we can fill the entries of all DP-cells.
There is one cell $(k,C)$ such that $C$ corresponds to the equivalence
class that contains $s'$. The entry of this cell is ``yes'' if
and only if $t_{1}$ can still break. 
\begin{theorem}
There is an algorithm with running time $n^{O(k^{4})}$ that decides
whether a given team $t_{1}$ can still break if there are at most
$k$ remaining rounds to play in a tournament, for an arbitrary breaking
threshold $b$ that is part of the input. 
\end{theorem}

{\bf Acknowledgments.}
The research leading to these results has received funding from the 
European Research Council under the European Union’s Seventh 
Framework Programme (FP/2007-2013) / ERC Grant Agreement no.\ 340506. 
Additionally, this house would like to thank Nicolas McLardy for hosting Stefan
in Berlin during the time when this paper was written.%

\sloppy
\bibliographystyle{plain}
\bibliography{Bib}

\end{document}